\documentclass[11pt]{article}
\usepackage{graphicx}
\usepackage{amsmath}
\usepackage{amssymb}
\usepackage{amsthm}
\usepackage{multirow}
\usepackage{subcaption}
\usepackage{amsfonts}
\usepackage{bbm}
\usepackage{xcolor}
\usepackage{url}
\usepackage{tabularx}
\newcolumntype{Y}{>{\centering\arraybackslash}X}
\usepackage[colorlinks=true]{hyperref}
\definecolor{darkred}  {rgb}{0.5,0,0}
\definecolor{darkblue} {rgb}{0,0,0.5}
\definecolor{darkgreen}{rgb}{0,0.5,0}
\hypersetup{
  urlcolor   = blue,         % color of external links
  linkcolor  = darkblue,     % color of internal links
  citecolor  = darkgreen,    % color of links to bibliography
  filecolor  = darkred       % color of file links
}
\usepackage[nameinlink]{cleveref}
\usepackage{caption}

\usepackage{tikz}
\usetikzlibrary{plotmarks}
\usetikzlibrary{shapes}
\usetikzlibrary{positioning}
\usetikzlibrary{calc}

\usetikzlibrary{external}
\tikzexternaldisable

\usepackage{etoolbox}
\newbool{FULL}
\booltrue{FULL}
\usetikzlibrary{quantikz2}

\newtheorem{defn}{Definition}

\newtheorem{lemma}{Lemma}
\newtheorem{theorem}{Theorem}
\newtheorem{corollary}{Corollary}

\newenvironment{myitemize}
{ \begin{itemize}
    \setlength{\itemsep}{5pt}
    \setlength{\parskip}{0pt}
    \setlength{\parsep}{0pt}     }
{ \end{itemize}   }

\def\ket#1{| #1 \rangle}

\ifbool{FULL}{
\title{Improved Clifford operations in constant commutative depth}
}{
\title{Improved Clifford operations in constant commutative depth \\
{\Large(extended abstract)}}
}
\author{Richard Cleve\footnote{Institute for Quantum Computing and School of Computer Science, University of Waterloo.} \and Zhiqian Ding$^{\ast}$ \and Luke Schaeffer$^{\ast}$}

\date{}

\begin{document}

\maketitle

\begin{abstract}
    The commutative depth model allows gates that commute with each other to be performed in parallel.
    We show how to compute Clifford operations in constant commutative depth more efficiently than was previously known.
    Bravyi, Maslov, and Nam [\textit{Phys.~Rev.~Lett.} 129:230501, 2022] showed that every element of the Clifford group (on $n$ qubits) can be computed in commutative depth 23 and size $O(n^2)$. We show that the Prefix Sum problem can be computed in commutative depth 16 and size $O(n \log n)$, improving on the previous depth 18 and size $O(n^2)$ bounds. We also show that, for arbitrary Cliffords, the commutative depth bound can be reduced to 16. Finally, we show some lower bounds: that there exist Cliffords whose commutative depth is at least 4; and that there exist Cliffords for which any constant commutative depth circuit has size $\Omega(n^2)$.
\end{abstract}
    
%%%%%%%%%%%%%%%%%%%%%%%%%%%%%%%%%%
\ifbool{FULL}{
\section{Introduction and summary of results}
\label{sec:intro}
}{
\section{Introduction and statement of results}\label{sec:intro}
}
%%%%%%%%%%%%%%%%%%%%%%%%%%%%%%%%%%

The standard notion of quantum circuit depth is based on the idea that any set of gates that act on distinct qubits can, in principle, be performed simultaneously in one parallel step.
We investigate a relaxation of this notion of depth, where we assume that any set of \emph{mutually commuting} gates can be performed in one parallel step.

The theoretical motivation for such a model is that if $m$ (possibly overlapping) unitary gates \linebreak 
$U_1 = e^{-iH_1t}, \dots, U_m = e^{-iH_mt}$ are mutually commuting then the Hamiltonians $H_1, \dots, H_m$ also commute and $U_1U_2\cdots U_m = e^{-i(H_1 + \cdots + H_m)t}$.
Therefore, one can, in principle, apply the processes associated with $H_1, \dots, H_m$ simultaneously to compute $U_1\cdots U_m$ in one parallel step.

In practice, many physical implementations of unitary operations $U$ are not as simple as applying a time-independent Hamiltonian acting on the same Hilbert space as $U$ for a fixed amount of time.
We do \emph{not} claim that if $U_1, \dots, U_m$ commute then every physical implementation of these gates can be performed in parallel.%
\footnote{In fact, for the standard notion of depth, implementing gates that act on separate qubits in parallel is nontrivial. See~\cite{FiggattO+2019} for a parallel implementation of gates acting on separate qubits. See~\cite{GrzesiakB+2020} for a parallel implementation of overlapping gates that commute.}
Rather, the theoretical existence of a set of natural commuting Hamiltonians is evidence that physical implementations that can be parallelized might be found.

The benefit of considering \emph{commutative depth} is that some computations require dramatically less commutative depth than standard depth (the reduction can be from $\Theta(n/\log n)$ to constant).
So, even if an implementation that takes advantage of commutative depth is more challenging than one that does not, the advantage of performing a computation in much fewer steps might justify the trouble of implementing commuting gates in parallel.
Related prior work includes H\o yer and \u{S}palek's reversible \emph{fanout} gate \cite{hoyer2005quantum} (see also \cite{takahashi2016collapse}) and the \emph{global tunable} gate (discussed in \cite{allcock2024constant}, with further references therein); each of these multi-qubit gates can be viewed as an arrangement of commuting 2-qubit gates.

\subsection{Every element of the Clifford group has constant commutative depth}

A remarkable result of Bravyi, Maslov and Nam \cite{BravyiMaslov+2022} implies that the commutative depth of each element of the Clifford group is at most 23 (with respect to the gate set 
$\{\mbox{\textsf{CNOT}},\mbox{\textsf{C}$Z$},\mbox{\textsf{C}$Y$}\}\cup\langle H,S\rangle$).

A drawback of this construction is the number of 2-qubit gates that occur.
The method in~\cite{BravyiMaslov+2022} converts Clifford group elements---some of which are computable with $O(n)$ gates---into circuits with constant commutative depth consisting of $\Theta(n^2)$ gates.
In general, constant commutative depth is attained at the cost of possibly increasing the total number of gates to $\Theta(n^2)$.
An interpretation of this is that the amount of ``work" involved in a parallel step may entail an amount of hardware that scales as $\Theta(n^2)$.

Our results are about efficiency improvements in constant commutative depth constructions of Clifford operations.
Note that we are considering \textit{in-place} circuits, that use no ancilla qubits. The computations are easier if one can employ many ancilla qubits, each initialized in state $\ket{0}$.

\subsection{New results about the Prefix Sum problem}

The \emph{Prefix Sum} problem is the problem of implementing the unitary operation $P$ on $n$ qubits satisfying
\begin{align}\label{eq:prefix-sum-intro}
    P\ket{b_1}\ket{b_2}\ket{b_3}\cdots\ket{b_n} = \ket{b_1}\ket{b_1 + b_2}\ket{b_1 + b_2 + b_3}\cdots\ket{b_1 + b_2 + b_3 + \cdots + b_n},
\end{align}
for all $b_1, b_2, \dots, b_n \in \{0,1\}$ (where additions are mod 2).
This is easily computed in depth $n-1$ as:
%\begin{align}\label{circuit:staircase}
%    \begin{matrix}\includegraphics[scale=0.17]{figures/StaircaseV4.png}\end{matrix}
%\end{align}
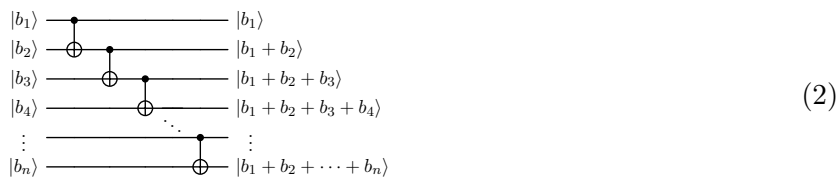
\tikzsetnextfilename{StaircaseV4}
\begin{align}\label{circuit:staircase}
	\begin{tikzpicture}[baseline=(current bounding box.center)]
		\node[scale=0.65] {
			\begin{quantikz}[row sep={6mm,between origins},column sep=4mm]
				\lstick{$\ket{b_1}$} & \ctrl{1} & & & & & \rstick{$\ket{b_1}$} \\
				\lstick{$\ket{b_2}$} & \targ{} & \ctrl{1} & & & & \rstick{$\ket{b_1 + b_2}$} \\
				\lstick{$\ket{b_3}$} & & \targ{} & \ctrl{1} & & & \rstick{$\ket{b_1 + b_2 + b_3}$} \\
				\lstick{$\ket{b_4}$} & & & \targ{} & \wire[l][1]["\ddots"{below=-1mm},xshift=2mm]{a} & & \rstick{$\ket{b_1 + b_2 + b_3 + b_4}$} \\
				\lstick{\vdots\ \ } & & & & & \ctrl{1} & \rstick{\ \ \vdots} \\
				\lstick{$\ket{b_n}$} & & & & & \targ{} & \rstick{$\ket{b_1 + b_2 + \cdots + b_n}$} \\
			\end{quantikz}
		};
	\end{tikzpicture}
\end{align}
It is well known that Prefix Sum can be computed in $\Theta(\log n)$ standard depth while preserving size $\Theta(n)$ by the method of Ladner and Fischer~\cite{LadnerFischer1980}.
Their construction is recursive and unfolds to consist of \textsf{CNOT} gates arranged in binary tree structures.
The rough idea is illustrated by this circuit for the $n = 8$~case (see also Eq.~\eqref{fig:LF-circuit} for the $n=16$ case):
%\begin{align}\label{circuit:LF-intro}
%\begin{matrix}\hspace*{5mm}\includegraphics[scale=0.17]{figures/LadnerFischer1980V2.png}\end{matrix}
%\caption{}{\mbox{Ladner-Fischer circuit}} \nonumber
%\end{align}
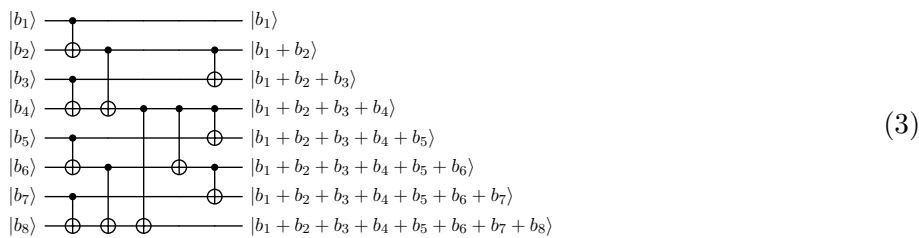
\tikzsetnextfilename{LadnerFischer1980V2}
\begin{align}\label{circuit:LF-intro}
	\begin{tikzpicture}[baseline=(current bounding box.center)]
	\node[scale=0.65] {
	\begin{quantikz}[row sep={6mm,between origins},column sep=4mm]
		\lstick{$\ket{b_1}$} & \ctrl{1} &          &          &          &          & \rstick{$\ket{b_1}$} \\
		\lstick{$\ket{b_2}$} & \targ{}  & \ctrl{2} &          &          & \ctrl{1} & \rstick{$\ket{b_1 + b_2}$} \\
		\lstick{$\ket{b_3}$} & \ctrl{1} &          &          &          & \targ{}  & \rstick{$\ket{b_1 + b_2 + b_3}$} \\
		\lstick{$\ket{b_4}$} & \targ{}  & \targ{}  & \ctrl{4} & \ctrl{2} & \ctrl{1} & \rstick{$\ket{b_1 + b_2 + b_3 + b_4}$} \\
		\lstick{$\ket{b_5}$} & \ctrl{1} &          &          &          & \targ{}  & \rstick{$\ket{b_1 + b_2 + b_3 + b_4 + b_5}$} \\
		\lstick{$\ket{b_6}$} & \targ{}  & \ctrl{2} &          & \targ{}  & \ctrl{1} & \rstick{$\ket{b_1 + b_2 + b_3 + b_4 + b_5 + b_6}$} \\
		\lstick{$\ket{b_7}$} & \ctrl{1} &          &          &          & \targ{}  & \rstick{$\ket{b_1 + b_2 + b_3 + b_4 + b_5 + b_6 + b_7}$} \\
		\lstick{$\ket{b_8}$} & \targ{}  & \targ{}  & \targ{}  &          &          & \rstick{$\ket{b_1 + b_2 + b_3 + b_4 + b_5 + b_6 + b_7 + b_8}$} \\
	\end{quantikz}
	};
	\end{tikzpicture}
\end{align}

%In addition to depth and size, there is a cost measure called the \emph{area} of a circuit.
%For the standard circuit depth model, the area can be defined as the number of layers times the number of qubits.
%This is relevant for fault-tolerant implementations, where every qubit has gates associated with it at each time step, whether or not a gate in the original circuit acts on it.
%For the commutative depth model, it is natural to define a weighted version of area (capturing the possibility of multiple commuting gates acting on a qubit in one time-step), based on defining the area of a layer as the maximum of $n$ and the number of gates acting on that layer---which can be asymptotically larger than $n$.

%The size (number of gates) of the circuit in Eq.~\eqref{circuit:staircase} is $n-1$ and the area of the circuit in 
%Eq.~\eqref{circuit:LF-intro} is $\Theta(n \log n)$.
The constructions in \cite{BravyiMaslov+2022} achieve commutative depth $18$ for Prefix sum, but at a cost of $\Theta(n^2)$ gates.
Note that the number of gates per parallel step is up to $\Theta(n^2)$.

Our first contribution shows how to attain constant commutative depth without the quadratic blow-up in size.
\begin{theorem}\label{thm:prefix-sum}
    For all even $n$, Prefix Sum can be computed in commutative depth 16 and size $\Theta(n \log n)$. (For odd $n$, the depth is 17.)
\end{theorem}
\noindent
Note that the number of gates per parallel step in this new construction is $O(n \log n)$.

Our methodology exploits structural properties of the Ladner-Fischer circuits.
Namely, that they can be decomposed into two parts, $L$ and $R$, that have these properties: (a) $L$ and $R$ are ``sparse" as linear operators; and (b) $R$ is equivalent to $H^{\otimes n} L^{-1} H^{\otimes n}$ if the order of qubits is reversed.

\medskip

\noindent
A summary of results for (in-place) computations of Prefix Sum is:\vspace*{1mm}

\begin{center}
    \renewcommand{\arraystretch}{1.1}
    \begin{tabular}{|l||c|c|c|}
    \hline
    Circuit construction & standard depth & commutative depth & size \\ \hline \hline
    Trivial, Eq.~\eqref{circuit:staircase} & $n-1$ & $n-1$ & $n-1$ \\ \hline
    Ladner-Fischer \cite{LadnerFischer1980}, Eqs.~\eqref{circuit:LF-intro}\eqref{fig:LF-circuit} & $\Theta(\log n)$ & $\Theta(\log n)$ & $\Theta(n)$ \\ \hline
    Bravyi \textit{et al.} \cite{BravyiMaslov+2022} & - & $18$ & $\Theta(n^2)$ \\ \hline
    \textbf{New construction} & - & $16$ & $\Theta(n \log n)$ \\ \hline
\end{tabular}
\end{center}

\subsection{New improvements to the Bravyi, Maslov, and Nam  \texorpdfstring{\cite{BravyiMaslov+2022} }{}construction}

For standard depth, Jiang et al.~\cite{jiang2020optimal} show that optimal depth of an arbitrary $n$-qubit Clifford operation is
$\Theta \! \left( \max \left\{ \log n,\, {n^2}/((n+m) \log (n+m)) \right\} \right)$, where 
$m$ is the number of ancilla qubits.
This is $\Theta(n / \log n)$ depth for the case of in-place circuits (i.e., with no ancillas).%
\footnote{In fact, the result in \cite{jiang2020optimal} implies that the depth is $\Theta(n / \log n)$ even if $O(n)$ ancilla qubits are permitted.}

The lower bounds for standard circuit depth come in two flavours: light cones and counting arguments.
Both of these are broken by the commutative depth model. A light cone lower bound argues that each input bit influences a bounded number of bits in the next layer, which influence a bounded number of bits in the next layer, and so on, and therefore functions which magnify the influence of a bit cannot be computed in low depth. Commutative depth breaks this by allowing layers where one bit influences many. The fan-out gate~\cite{hoyer2005quantum}, for example, requires standard depth $\log n$, but has commutative depth $1$.  

Counting arguments are based on the principle that there are a finite number of choices for each layer in a circuit, and this bounds the number of distinct circuits of a given depth. This must be at least the number of functions we wish to compute, otherwise there must be \emph{some} circuit we cannot compute. Commutative depth significantly weakens counting arguments, since vastly more commuting layers of gates are possible compared to disjoint layers.

%We contribute our own lower bounds, based on counting arguments, % in  \Cref{thm:lower_bound_linear} and \Cref{thm:lower_bound_clifford},
%but they are limited to the specific constants in commutative depth.

The remarkable result of Bravyi, Maslov and Nam \cite{BravyiMaslov+2022} shows that any Clifford operation can be implemented in constant commutative depth.
In their model, the commuting layers (which they call ``GCZ gates") are equivalent to collections of 2-qubit C$Z$ gates.
%A layer of commuting \textsf{CNOT} gates is implemented in their model by applying the same set of C$Z$ gates and then conjugating the target bits by $H$ gates.
They use $20$ layers%
 \footnote{Technically, Bravyi \textit{et al.} use an extra layer when $n$ is not divisible by three. We have a similar issue with parity. Since we quote the lower number (for even $n$) in our results, we extend their result the same courtesy. We also note that one can also ``pad'' the transformation by taking the direct sum with an identity ($I_1$ or $I_2$) to fix the divisibility at the cost of one or two dirty ancillas.} of GCZ gates. However, they use a decomposition \cite{Bravyi2020HadamardFreeCE} which has three additional layers of single-qubit gates, and thus the commutative depth is $23$.

Our contribution here is an improved construction, where every element of the Clifford group can be computed in commutative depth $16$ with $\Theta(n^2)$ gates. The key underlying idea is a better implementation of a certain group commutator. As part of the construction, we also get circuits of commutative depth $11$ for linear operations on $n$ qubits (for the case of even $n$).

Finally, we show two lower bounds.
First, that the commutative depth of an arbitrary Clifford operation is at least 4.
Second, there exists a Clifford operation for which any constant commutative depth circuit computing it has size $\Omega(n^2)$. This is noteworthy because it is larger than the size without any depth restriction, which is $O(n^2/\log )$~\cite{AaronsonGottesman2004}.

%The approach in \cite{BravyiMaslov+2022} achieves a commutative circuit depth of $23$ when $3|n$ and at most $24$ for all Clifford operations. Our main contribution is demonstrating that every element of the Clifford group can be computed in depth 13 with $O(n^2)$ gates when $2|n$.

%Our main contributions here is that every element of the Clifford group can be computed in depth 13 with $O(n^2)$ gates.

\medskip

\noindent
A summary of results for (in-place) computations of arbitrary Clifford operations is:\vspace*{1mm}

\begin{center}
\renewcommand{\arraystretch}{1.2}
\begin{tabular}{|l||c|c|c|}
    \hline
    Circuit construction & standard depth & \!commutative depth\! & size \\ \hline \hline
    Aaronson et al. \cite{AaronsonGottesman2004} / Jiang et al. \cite{jiang2020optimal} & $\Theta(\frac{n}{\log n})$ & $\Theta(\frac{n}{\log n})$ & $\Theta(\frac{n^2}{\log n})$ \\ \hline
    Bravyi \textit{et al.} \cite{BravyiMaslov+2022} & - & $23$ & $O(n^2)$ \\ \hline
    \textbf{New construction} & - & $16$ & $O(n^2)$ \\ \hline
    \textbf{New depth lower bound} & - & $\ge 4$ & - \\ \hline
    \textbf{New size lower bound for depth $O(1)$} & - & $O(1)$ & $\Omega(n^2)$ \\ \hline
\end{tabular}
\end{center}

\section{Definitions and notation}

\subsection{Definition of commutative depth}

Assume that we have a fixed generating set of quantum gates.
For the case of Clifford operations, a reasonable set of gates that generate them is 
$\{\mbox{\textsf{CNOT}},\mbox{\textsf{C}$Z$},\mbox{\textsf{C}$Y$}\}\cup\langle H,S\rangle$ (the set of all 1-qubit Clifford gates as well as controlled-Pauli gates).
Intuitively, performing one gate corresponds to a ``constant amount of work". 

\begin{defn}\label{def:commutative-depth}
Define a \emph{commuting layer} as a quantum circuit consisting gates from the generating set 
such that every pair of gates in the circuit commute.
Define a \emph{layered circuit} as a circuit that is a composition of commuting layers.
The \emph{commutative depth} of such a circuit is the number of commuting layers.
\end{defn}

The idea is that, when gates commute, their canonical Hamiltonians also commute, which implies that, in principle, all the gates in such a layer can be performed simultaneously in one step
(this is discussed in more detail in the first four paragraphs of section~\ref{sec:intro}).

\medskip

\noindent
\textbf{Notes:}\vspace*{-3mm}

\begin{myitemize}
    \item
    The \emph{standard depth} can be defined similarly to Definition~\ref{def:commutative-depth} with a more stringent definition of a layer: where all the gates in a layer are required to act on distinct qubits.
    \item
    We are considering \textit{in-place} circuits, that use no ancilla qubits. The computations are easier if one can employ many ancilla qubits, each initialized in state $\ket{0}$.
    \item
    The \emph{size} of a layered quantum circuit is defined as the total number of gates it contains.
\end{myitemize}

%We also, define the area of a quantum circuit, which is relevant for fault-tolerant implementations, where every qubit has gates associated with it, even if no gate in the original circuit acts on it.
%For the commutative depth model, it is natural to define the area of a layer is the maximum of $n$ and the number of gates acting on that layer---which can be asymptotically larger than $n$.
%The area of a layered circuit is the sum of the areas of all its layers.

%\begin{defn}
%    The \emph{area} of a commuting layer is the maximum of $n$ and the number of gates acting on that layer (which can be asymptotically larger than $n$).
%    The \emph{area} of a layered circuit is the sum of the areas of all its layers.
%\end{defn}

\subsection{Linear permutations}

\begin{defn}
    Define a \emph{linear permutation} on $n$ qubits as a unitary operation that permutes the computational basis states as
    \begin{align}
    \ket{x}\ \mapsto \ket{Mx},
    \end{align}
    for all $x \in \{0,1\}^n$,
    where $M$ is some invertible $n \times n$ binary matrix, and $Mx$ denotes left multiplication by $M$ of $x$ as a column vector (in mod $2$ arithmetic). 
\end{defn}

\noindent
Circuit notation for such a linear permutation is the following (where the wire denotes $n$ qubits):
\tikzset{operator/.append style={text height=1.7em}}
\begin{align}
%    \begin{matrix}
%        \includegraphics[scale=0.15]{figures/M.png}
%    \end{matrix}
    \tikzsetnextfilename{M}
    \begin{quantikz}    
    	& \gate{M} & 
    \end{quantikz}
\end{align}

For example, the Prefix Sum $P$ defined in Eq.~\eqref{eq:prefix-sum-intro} is a linear permutation with associated matrix
\begin{align}\label{eq:P}
    P = \begin{bmatrix}
        1 & 0 & 0 & \cdots & 0 \\
        1 & 1 & 0 & \cdots & 0\\
        1 & 1 & 1 & \cdots & 0\\
        \vdots & \vdots & \vdots & \ddots & \vdots \\
        1 & 1 & 1 & \cdots  & 1 \\
    \end{bmatrix}.
\end{align}

\subsection{Linear additions}

\begin{defn}
For an arbitrary binary $n \times n$ matrix $M$ (not necessarily invertible), the transformation on two $n$-qubit registers 
\begin{align}\label{eq:M-XOR} 
\ket{x}\ket{y} & \mapsto \ket{x}\ket{y+Mx}
\end{align}
is the linear permutation that corresponds to the $2n \times 2n$ matrix
\begin{align}
    \begin{bmatrix}
        I & 0 \\
        M & I
    \end{bmatrix}.
%    \hspace*{10mm}\mbox{and}\hspace*{10mm}
%    \begin{bmatrix}
%        I & M \\
%        0 & I
%    \end{bmatrix}.
\end{align}
\end{defn}
\noindent
Our circuit notation for this is the following (where each wire denotes $n$ qubits):
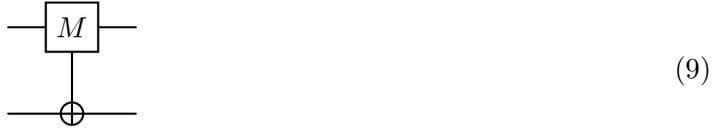
\begin{align}
%    \begin{matrix}
%        \includegraphics[scale=0.15]{figures/M-XOR.png}
%    \end{matrix}
    \tikzsetnextfilename{M-XOR}
    \begin{quantikz}[row sep={11.5mm,between origins},align equals at=1.5]
    	& \gate{M} & \\
    	& \targ{}\wire[u][1]{q} &
    \end{quantikz}
\end{align}
\noindent
Note that this operation is implementable by one single commuting layer of $w$ \textsf{CNOT} gates, 
where $w = \mbox{weight}(M)$, which is the number of non-zero entries of $M$.
This is because a direct implementation of the circuit in terms of \textsf{CNOT} gates has one such gate for each non-zero entry of $M$ with control-qubit among the first $n$ qubits and target qubit among the last $n$ qubits---hence all these \textsf{CNOT} gates commute.

Also, note that two linear additions can be combined when their control and target qubits are aligned, as:
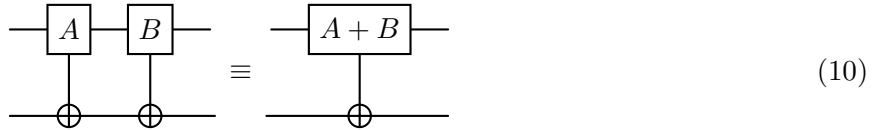
\begin{align}\label{fig:R+L}
%    \begin{matrix}
%        \includegraphics[scale=0.26]{figures/R+L.png}
%    \end{matrix} \\
    \tikzsetnextfilename{R+L}
    \begin{quantikz}[row sep={11.5mm,between origins},align equals at=1.5]
    	& \gate{A} & \gate{B} & \midstick[wires=2,brackets=none]{$\equiv$} & \gate{A+B} & \\
    	& \targ{}\wire[u][1]{q} & \targ{}\wire[u][1]{q} &  & \targ{}\wire[u][1]{q} &
    \end{quantikz} 
\end{align}

%------------------------------------------------------------------------------%
\subsection{The \texorpdfstring{$M\oplus M^{-1}$}{M and M inverse} mapping in constant commutative depth}

For any two invertible $n\times n$ matrices $A$ and $B$, we use the notation $A\oplus B$ 
to refer to the $2n$-qubit linear permutation corresponding to the direct sum of $A$ and $B$, namely
\begin{align}
 A \oplus B = 
    \begin{bmatrix}
        A & 0 \\
        0 & B
    \end{bmatrix}.
\end{align}

It is straightforward to verify that, for any invertible $n \times n$ binary matrix $M$, the $M \oplus M^{-1}$ and $M^{-1} \oplus M$ mappings are computed, respectively, by the circuits in Eq.~\eqref{fig:M-M}.

%\begin{align}\label{fig:M-M}
%    \begin{matrix}
%        \includegraphics[scale=0.26]{figures/M-M-inv-new.png}\hspace*{15mm}
%        \includegraphics[scale=0.26]{figures/M-inv-M-new.png} \\
%        \mbox{(a)}\hspace*{80mm}
%        \mbox{(b)}
%    \end{matrix}
%\end{align}
\begin{align}\label{fig:M-M}
	\begin{tabular}{cc}
		\tikzsetnextfilename{M-M-inv-new}
		\begin{quantikz}[row sep={11.5mm,between origins},align equals at=1.5,column sep=2mm]
			& \gate{M} & \targ{}\wire[d][1]{q} & \gate{M} & \permute{2,1} & \midstick[wires=2,brackets=none]{$\equiv$} & \gate{M} & \\
			& \targ{}\wire[u][1]{q} & \gate{M^{-1}} & \targ{}\wire[u][1]{q} & & & \gate{M^{-1}} &
		\end{quantikz} \ \ 
		& \ \ 
		\tikzsetnextfilename{M-inv-M-new}
		\begin{quantikz}[row sep={11.5mm,between origins},align equals at=1.5,column sep=2mm]
			& \permute{2,1} & \gate{M} & \targ{}\wire[d][1]{q} & \gate{M} & \midstick[wires=2,brackets=none]{$\equiv$} & \gate{M^{-1}} & \\
			& & \targ{}\wire[u][1]{q} & \gate{M^{-1}} & \targ{}\wire[u][1]{q} & & \gate{M} &
		\end{quantikz}
		\\
		(a) & (b) \\
	\end{tabular}
\end{align}

\noindent
The commutative depth of these circuits is 3 if we do not count the swap gates (and it turns out that, in our construction, there are swap gates that cancel out and therefore can be eliminated).

%------------------------------------------------------------------------------%
\section{New results for the Prefix Sum problem}\label{sec:prefix-sums}

The prefix sum problem corresponds to the unitary $P$ as defined in Eq.~\eqref{eq:prefix-sum-intro}.
As explained in section~\ref{sec:intro}, previous methods for computing this in place either have (standard) depth $\Theta(\log n)$ and area $\Theta(n \log n)$ or commutative depth constant and size $\Theta(n^2)$ (which is also the area).

In this section, we prove Theorem~\ref{thm:prefix-sum}, which states that constant commutative depth can be attained while preserving size (and weighted area) $\Theta(n \log n)$.

%We give a more efficient Clifford circuit for $P$ which consists of 16 layers, with only $\Theta(n \log n)$ commuting \textsf{CNOT} and/or Hadamard gates per layer.

%------------------------------------------------------------------------------%
%\section{New \texorpdfstring{$O(1)$}{O(1)}-depth and \texorpdfstring{$O(n \log n)$}{O(n log n)}-size circuit construction}

%------------------------------------------------------------------------------%
\subsection{Structural symmetries of the Ladner-Fischer circuit for prefix sums}

The starting point of our construction is the elegant parallel algorithm for computing prefix sums due to Ladner and Fischer \cite{LadnerFischer1980}, which shows how to compute the mapping 
\begin{align}
    (x_1, x_2, \dots, x_n) \mapsto (x_1,\, x_1 \!\circ\! x_2,\, \dots,\, x_1\!\circ \cdots \circ\!x_n),
\end{align}
with respect to any associative operation $\circ$ on some domain, with $O(n)$ $\circ$ gates in depth $O(\log n)$.
A special case of this is where the domain is $\{0,1\}$ and the binary operation is addition modulo 2.

When $n=16$ the Ladner-Fischer circuit looks like this:
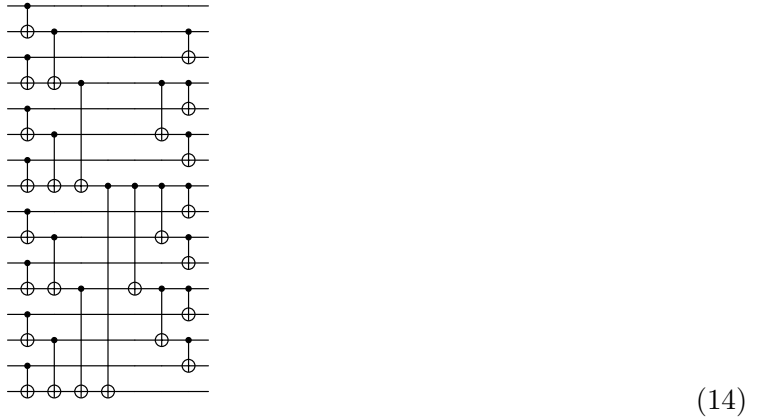
\begin{align}\label{fig:LF-circuit}
%    \begin{matrix}
%        \includegraphics[scale=0.20]{figures/LadnerFischer1980.png} 
%        \mbox{Ladner-Fischer circuit (illustrated for $n = 16$)}
%    \end{matrix} \\
    \tikzsetnextfilename{LadnerFischer1980}
    \begin{tikzpicture}
    \node[scale=0.57] {
    	 \begin{quantikz}[row sep={6mm,between origins},column sep=3mm,align equals at=8]
    		& \ctrl{1} & & & & & & & \\
    		& \targ{} & \ctrl{2} & & & & & \ctrl{1} & \\
    		& \ctrl{1} & & & & & & \targ{} & \\
    		& \targ{} & \targ{} & \ctrl{4} & & & \ctrl{2} & \ctrl{1} & \\
    		& \ctrl{1} & & & & & & \targ{} & \\
    		& \targ{} & \ctrl{2} & & & & \targ{} & \ctrl{1} & \\
    		& \ctrl{1} & & & & & & \targ{} & \\
    		& \targ{} & \targ{} & \targ{} & \ctrl{8} & \ctrl{4} & \ctrl{2} & \ctrl{1} & \\
    		& \ctrl{1} & & & & & & \targ{} & \\
    		& \targ{} & \ctrl{2} & & & & \targ{} & \ctrl{1} & \\
    		& \ctrl{1} & & & & & & \targ{} & \\
    		& \targ{} & \targ{} & \ctrl{4} & & \targ{} & \ctrl{2} & \ctrl{1} & \\
    		& \ctrl{1} & & & & & & \targ{} & \\
    		& \targ{} & \ctrl{2} & & & & \targ{} & \ctrl{1} & \\
    		& \ctrl{1} & & & & & & \targ{} & \\
    		& \targ{} & \targ{} & \targ{} & \targ{}& & & & 
    	\end{quantikz}
    };
    \end{tikzpicture}
\end{align}

\noindent
More generally, for $n = 2^k$, the Ladner-Fischer circuit begins with \textsf{CNOT} gates arranged as a binary tree of depth $k$ rooted at the last qubit, followed by parallel composition of binary trees of \textsf{CNOT} gates oriented a different way with depths $k-1, k-2, \dots, 1$.
%This is the result of unravelling the recursive algorithm  in~\cite{LadnerFischer1980}.
These circuits compute the prefix sums in standard depth $O(\log n)$ and size $O(n)$, which is already a significant improvement over the circuit in Eq.~\eqref{circuit:staircase}.
However, our goal is to compute the prefix sums in \emph{constant commutative depth}, and this construction does not achieve that (no matter how one partitions the gates into commuting layers).

We obtain our constant-depth circuit construction by leveraging properties of the left and right parts of the circuit in Eq.~\eqref{fig:LF-circuit}.
To make the construction work cleanly, we \emph{prune}%
\footnote{\emph{Pruning} a qubit from a circuit, means removing that qubit as well as all gates incident with that qubit.}
the last qubit from the circuit, resulting a circuit on $2^k-1$ qubits of the form of circuit (a) in Eq.~\eqref{fig:LF-L-R}:
%\begin{align}\label{fig:LF-L-R}
%\begin{matrix}
%     \includegraphics[scale=0.2]{figures/PrunedLF.png}\hspace*{25mm}
%     \includegraphics[scale=0.2]{figures/L.png}\hspace*{25mm}
%     \includegraphics[scale=0.174]{figures/R.png} \\
%     \hspace*{-20mm}\mbox{(a) Pruned Ladner-Fischer circuit}\hspace*{15mm}
%     \mbox{(b) $L$}\hspace*{32mm}\mbox{(c) $R$}
%\end{matrix}
%\end{align}
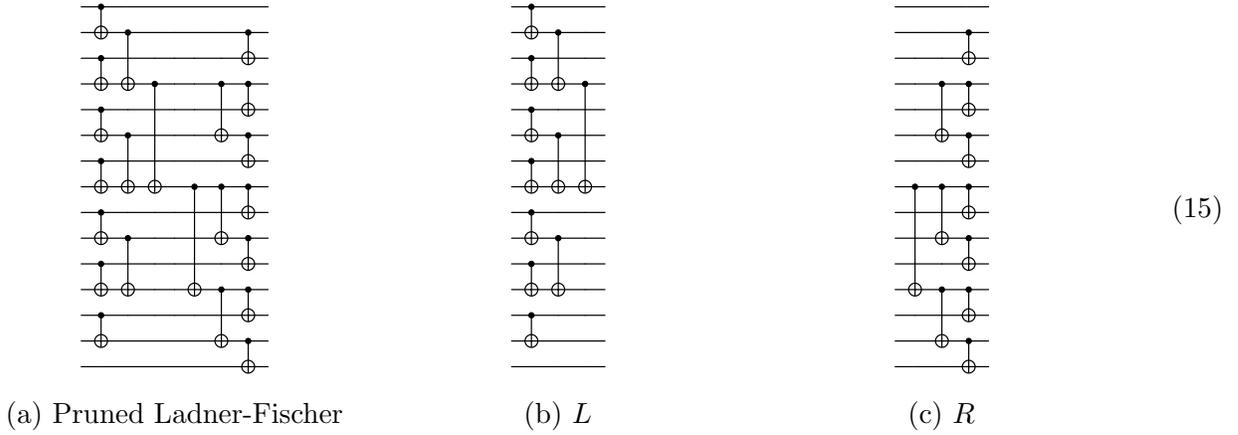
\begin{equation}\label{fig:LF-L-R}
\begin{tabularx}{.9\textwidth}{YYY}
	\tikzsetnextfilename{PrunedLF}
	\begin{tikzpicture}
		\node[scale=0.57] {
			\begin{quantikz}[row sep={6mm,between origins},column sep=3mm,align equals at=8]
				& \ctrl{1} & \phantomgate{;} & & & & & & \\
				& \targ{} & \ctrl{2} & & & & & \ctrl{1} & \\
				& \ctrl{1} & & & & & & \targ{} & \\
				& \targ{} & \targ{} & \ctrl{4} & & & \ctrl{2} & \ctrl{1} & \\
				& \ctrl{1} & & & & & & \targ{} & \\
				& \targ{} & \ctrl{2} & & & & \targ{} & \ctrl{1} & \\
				& \ctrl{1} & & & & & & \targ{} & \\
				& \targ{} & \targ{} & \targ{} & & \ctrl{4} & \ctrl{2} & \ctrl{1} & \\
				& \ctrl{1} & & & & & & \targ{} & \\
				& \targ{} & \ctrl{2} & & & & \targ{} & \ctrl{1} & \\
				& \ctrl{1} & & & & & & \targ{} & \\
				& \targ{} & \targ{} & & & \targ{} & \ctrl{2} & \ctrl{1} & \\
				& \ctrl{1} & & & & & & \targ{} & \\
				& \targ{} & & & & & \targ{} & \ctrl{1} & \\
				& \phantomgate{;} & & & & & & \targ{} & 
			\end{quantikz}
		};
	\end{tikzpicture} & 
	\tikzsetnextfilename{L}
	\begin{tikzpicture}
		\node[scale=0.57,anchor=base] {
			\begin{quantikz}[row sep={6mm,between origins},column sep=3mm,align equals at=1]
				& \ctrl{1} & &\phantomgate{;} & \\
				& \targ{} & \ctrl{2} & & \\
				& \ctrl{1} & & & \\
				& \targ{} & \targ{} & \ctrl{4} & \\
				& \ctrl{1} & & & \\
				& \targ{} & \ctrl{2} & & \\
				& \ctrl{1} & & & \\
				& \targ{} & \targ{} & \targ{} & \\
				& \ctrl{1} & & & \\
				& \targ{} & \ctrl{2} & & \\
				& \ctrl{1} & & & \\
				& \targ{} & \targ{} & & \\
				& \ctrl{1} & & & \\
				& \targ{} & & & \\
				& \phantomgate{;} & & & 
			\end{quantikz}
		};
	\end{tikzpicture} &
	\tikzsetnextfilename{R}
	\begin{tikzpicture}
		\node[scale=0.57,anchor=base] {
			\begin{quantikz}[row sep={6mm,between origins},column sep=3mm,align equals at=1]
				& & & \phantomgate{;} & \\
				& & & \ctrl{1} & \\
				& & & \targ{} & \\
				& & \ctrl{2} & \ctrl{1} & \\
				& & & \targ{} & \\
				& & \targ{} & \ctrl{1} & \\
				& & & \targ{} & \\
				& \ctrl{4} & \ctrl{2} & \ctrl{1} & \\
				& & & \targ{} & \\
				& & \targ{} & \ctrl{1} & \\
				& & & \targ{} & \\
				& \targ{} & \ctrl{2} & \ctrl{1} & \\
				& & & \targ{} & \\
				& & \targ{} & \ctrl{1} & \\
				& \phantomgate{;} & & \targ{} & 
			\end{quantikz}
		};
	\end{tikzpicture} \\
    (a) Pruned Ladner-Fischer & (b) $L$ & (c) $R$ 
\end{tabularx}
\end{equation}
\noindent
It is straightforward to check that this circuit correctly computes the parallel prefixes for any $n = 2^k -1$ qubits.
And this circuit is the composition of the circuits $L$ and $R$, shown respectively in parts (b) and (c) of Eq.~\eqref{fig:LF-L-R}, where the symmetry between $L$ and $R$ is easy to visualize.

A circuit for $L^{-1}$ consists of the gates of $L$ in reverse order.
What is the relationship between $L^{-1}$ and $R$?
$R$ is like an upside-down version of $L^{-1}$, with the additional change that each \textsf{CNOT} gate is inverted (in the sense of the control and target qubits being swapped; which is achieved by conjugating each qubit with $H$ gates).

To define the upside-down version of a circuit, let $B$ be the unitary operation that outputs the input qubits in backwards order (i.e., maps $\ket{b_1, b_2, \cdots, b_{n-1}, b_n}$ to $\ket{b_n, b_{n-1}, \cdots, b_2, b_1}$).
Then $BMB$ is the matrix associated with the upside-down version of a circuit for $M$, and the relationship between $L^{-1}$ and $R$ is illustrated in Eq.~\eqref{fig:3-circuits} (for $n=15$):
%\begin{align}\label{fig:3-circuits}
%    \begin{matrix}
%        \includegraphics[scale=0.2]{figures/L-inv.png}\hspace*{25mm}
%        \includegraphics[scale=0.2]{figures/L-inv-B.png}\hspace*{25mm}
%        \includegraphics[scale=0.174]{figures/R.png} \\
%        \hspace*{19mm}\mbox{(a) $L^{-1}$}\hspace*{24mm}
%     \mbox{(b) $BL^{-1}B$}\hspace*{15mm}\mbox{(c) $H^{\otimes n}BL^{-1}BH^{\otimes n} = R$}
%    \end{matrix}
%\end{align}

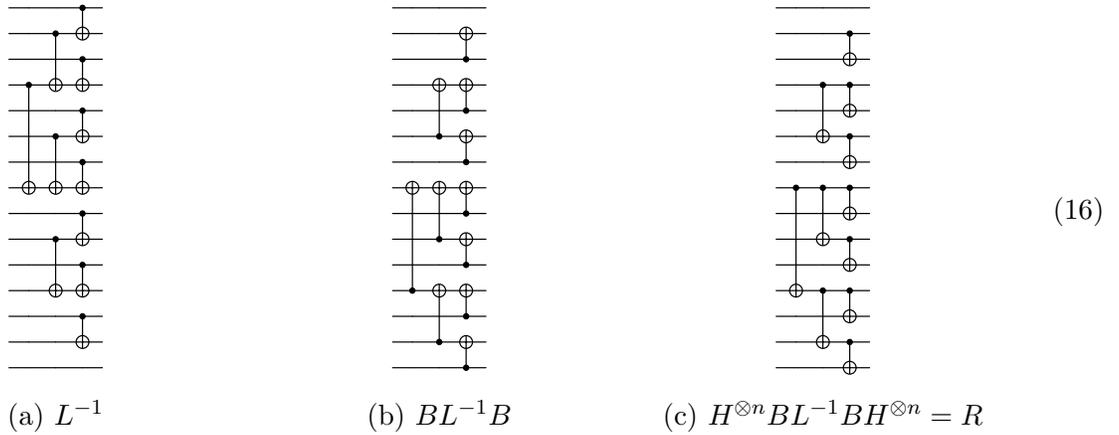
\begin{equation}\label{fig:3-circuits}
\begin{tabularx}{.9\textwidth}{YYY}
    \tikzsetnextfilename{Linv}
	\begin{tikzpicture}
		\node[scale=0.57,anchor=base] {
			\begin{quantikz}[row sep={6mm,between origins},column sep=3mm,align equals at=1]
				& & \phantomgate{;} & \ctrl{1} & \\
				& & \ctrl{2} & \targ{} & \\
				& & & \ctrl{1} & \\
				& \ctrl{4} & \targ{} & \targ{} & \\
				& & & \ctrl{1} & \\
				& & \ctrl{2} & \targ{} & \\
				& & & \ctrl{1} & \\
				& \targ{} & \targ{} & \targ{} & \\
				& & & \ctrl{1} & \\
				& & \ctrl{2} & \targ{} & \\
				& & & \ctrl{1} & \\
				& & \targ{} & \targ{} & \\
				& & & \ctrl{1} & \\
				& & & \targ{} & \\
				& \phantomgate{;} & & & 
			\end{quantikz}
		};
	\end{tikzpicture} &
	\tikzsetnextfilename{BLinvB}
	\begin{tikzpicture}
		\node[scale=0.57,anchor=base] {
			\begin{quantikz}[row sep={6mm,between origins},column sep=3mm,align equals at=1]
				& & & \phantomgate{;} & \\
				& & & \targ{} & \\
				& & & \ctrl{-1} & \\
				& & \targ{} & \targ{} & \\
				& & & \ctrl{-1} & \\
				& & \ctrl{-2} & \targ{} & \\
				& & & \ctrl{-1} & \\
				& \targ{} & \targ{} & \targ{} & \\
				& & & \ctrl{-1} & \\
				& & \ctrl{-2} & \targ{} & \\
				& & & \ctrl{-1} & \\
				& \ctrl{-4} & \targ{} & \targ{} & \\
				& & & \ctrl{-1} & \\
				& & \ctrl{-2} & \targ{} & \\
				& \phantomgate{;} & & \ctrl{-1} & 
			\end{quantikz}
		};
	\end{tikzpicture} &
	\tikzsetnextfilename{R-2}
	\begin{tikzpicture}
		\node[scale=0.57,anchor=base] {
			\begin{quantikz}[row sep={6mm,between origins},column sep=3mm,align equals at=1]
				& & & \phantomgate{;} & \\
				& & & \ctrl{1} & \\
				& & & \targ{} & \\
				& & \ctrl{2} & \ctrl{1} & \\
				& & & \targ{} & \\
				& & \targ{} & \ctrl{1} & \\
				& & & \targ{} & \\
				& \ctrl{4} & \ctrl{2} & \ctrl{1} & \\
				& & & \targ{} & \\
				& & \targ{} & \ctrl{1} & \\
				& & & \targ{} & \\
				& \targ{} & \ctrl{2} & \ctrl{1} & \\
				& & & \targ{} & \\
				& & \targ{} & \ctrl{1} & \\
				& \phantomgate{;} & & \targ{} & 
			\end{quantikz}
		};
	\end{tikzpicture} \\
    (a) $L^{-1}$ & (b) $BL^{-1}B$ & (c) $H^{\otimes n} BL^{-1}BH^{\otimes n} = R$   
\end{tabularx}
\end{equation}
\noindent
In summary, we have the following.
\begin{lemma}
    For $n = 2^k -1$, the following relationships between $L$ and $R$ hold:
\begin{align}
    H^{\otimes n}BL^{-1}BH^{\otimes n} &= R \label{eq:R-in-terms-of-L}\\
    H^{\otimes n}BR^{-1}BH^{\otimes n} &= L.\label{eq:L-in-terms-of-R}
\end{align}
\end{lemma}

%    ~
%    \begin{subfigure}[c]{0.25\textwidth}
%        \centering
%        \includegraphics[scale=0.174]{figures/R.png}
%        \caption{$H^{\otimes n}BL^{-1}BH^{\otimes n}$}
%    \end{subfigure}
%------------------------------------------------------------------------------%
%\subsection{Binary linear operations on computational basis states}

Our construction will use $L$, $R$, $L^{-1}$, and $R^{-1}$, and is gate-efficient because the these matrices have weight (i.e., density of 1s) close to linear.
In Appendix~\ref{appendix:A}, we prove the following.
\begin{lemma}
    The weight of the matrices corresponding to $L$ and $R$ are both $O(n \log n)$.
    The weight of the matrices corresponding to $L^{-1}$, and $R^{-1}$ are both $O(n)$.
\end{lemma}

%------------------------------------------------------------------------------%
\subsection{The \texorpdfstring{$P \oplus P$}{PP} mapping in commutative depth 15}

For $P$ defined in Eq.~\eqref{eq:P}, using the fact that $P = RL$, 
we can compute the linear permutation $P \oplus I$ by the circuit in Eq.~\eqref{fig:P-I}:
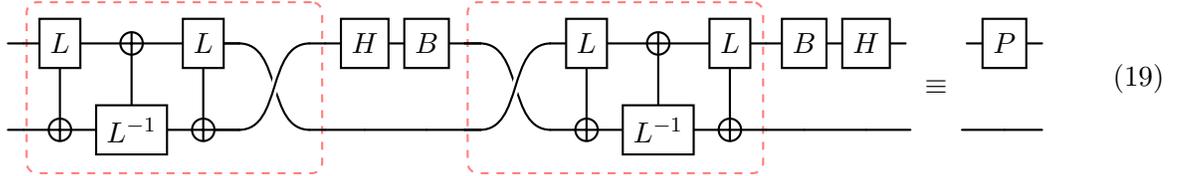
\begin{align}\label{fig:P-I}
%\begin{matrix}
%    \includegraphics[scale=0.26]{figures/P-I-new.png}
%\end{matrix} \\
\tikzsetnextfilename{P-I-new}
\begin{matrix}
\begin{quantikz}[row sep={11.5mm,between origins},align equals at=1.5,column sep=2.1mm]
	& & \gate{L}\gategroup[2,steps=4,style={dashed,rounded
		corners,draw=red!50,inner xsep=1pt},background]{} & \targ{}\wire[d][1]{q} & \gate{L} & \permute{2,1} & & \gate{H} & \gate{B} & & \permute{2,1}\gategroup[2,steps=4,style={dashed,rounded
		corners,draw=red!50,inner xsep=1pt},background]{} & \gate{L} & \targ{}\wire[d][1]{q} & \gate{L} & & \gate{B} & \gate{H} & \midstick[wires=2,brackets=none]{$\equiv$} & \gate{P} & \\
	& & \targ{}\wire[u][1]{q} & \gate{L^{-1}} & \targ{}\wire[u][1]{q} & & & & & & &\targ{}\wire[u][1]{q} & \gate{L^{-1}} & \targ{}\wire[u][1]{q} & & & & & &
\end{quantikz}
\end{matrix}
\end{align}

\noindent
This works because, due to Eq.~\eqref{fig:M-M}, it is equivalent to the circuit on the left side of Eq.~\eqref{fig:P-I-detail}, which reduces to $(RL)\oplus I$ by Eq.~\eqref{eq:R-in-terms-of-L}.
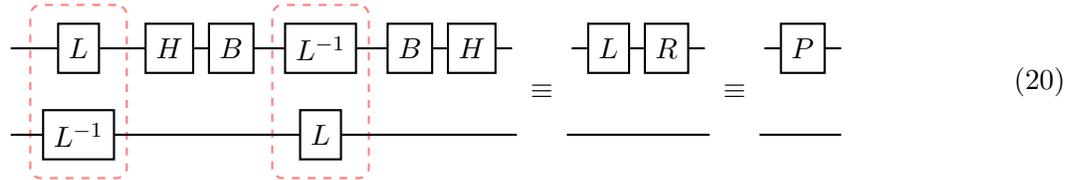
\begin{align}\label{fig:P-I-detail}
%    \begin{matrix}
%        \includegraphics[scale=0.26]{figures/P-I-detail.png}
%    \end{matrix} \\
    \tikzsetnextfilename{P-I-detail}
    \begin{matrix}
    	\begin{quantikz}[row sep={11.5mm,between origins},align equals at=1.5,column sep=2.1mm]
    		& & \gate{L}\gategroup[2,steps=1,style={dashed,rounded
    			corners,draw=red!50,inner xsep=1pt},background]{} & & \gate{H} & \gate{B} & & \gate{L^{-1}}\gategroup[2,steps=1,style={dashed,rounded
    			corners,draw=red!50,inner xsep=1pt},background]{} & & \gate{B} & \gate{H} & \midstick[wires=2,brackets=none]{$\equiv$} & \gate{L} & \gate{R} & \midstick[wires=2,brackets=none]{$\equiv$} & \gate{P} & \\
    		& & \gate{L^{-1}} & & & & &\gate{L} & & & & & & & & &
    	\end{quantikz}
    \end{matrix}
\end{align}

To compute $I\oplus P$, we use a variant of the above construction, shown in Eq.~\eqref{fig:I-P}:
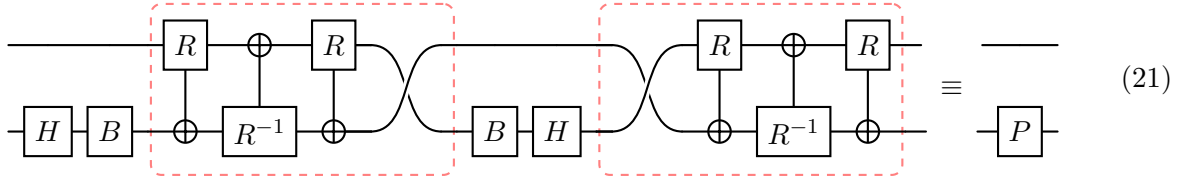
\begin{align}\label{fig:I-P}
%    \begin{matrix}
%        \includegraphics[scale=0.26]{figures/I-P-new.png}
%    \end{matrix}  \\
    \tikzsetnextfilename{I-P-new}
    \begin{matrix}
    \begin{quantikz}[row sep={11.5mm,between origins},align equals at=1.5,column sep=2.1mm]
    	& & & & \gate{R}\gategroup[2,steps=4,style={dashed,rounded
    		corners,draw=red!50,inner xsep=1pt},background]{} & \targ{}\wire[d][1]{q} & \gate{R} & \permute{2,1} & & & & & \permute{2,1}\gategroup[2,steps=4,style={dashed,rounded
    		corners,draw=red!50,inner xsep=1pt},background]{} & \gate{R} & \targ{}\wire[d][1]{q} & \gate{R} & & \midstick[wires=2,brackets=none]{$\equiv$} & & \\
    	& \gate{H} & \gate{B} & & \targ{}\wire[u][1]{q} & \gate{R^{-1}} & \targ{}\wire[u][1]{q} & & & \gate{B} & \gate{H} & & &\targ{}\wire[u][1]{q} & \gate{R^{-1}} & \targ{}\wire[u][1]{q} & & & \gate{P} &
    \end{quantikz}
    \end{matrix}
\end{align}

\noindent
This works because, due to Eq.~\eqref{fig:M-M}, it is equivalent to the circuit on the left side of Eq.~\eqref{fig:I-P-detail}, which reduces to $I\oplus (RL)$ by Eq.~\eqref{eq:L-in-terms-of-R}.
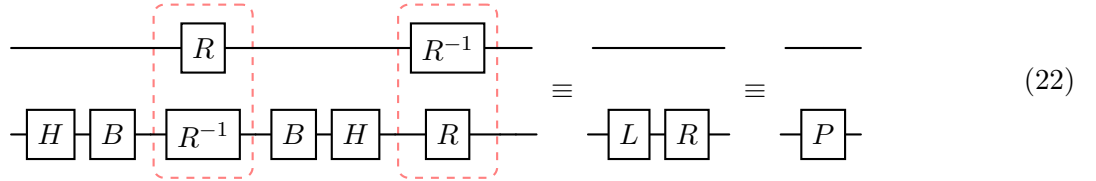
\begin{align}\label{fig:I-P-detail}
%    \begin{matrix}
%        \includegraphics[scale=0.26]{figures/I-P-detail.png}
%    \end{matrix} \\
    \tikzsetnextfilename{I-P-detail}
    \begin{matrix}
    \begin{quantikz}[row sep={11.5mm,between origins},align equals at=1.5,column sep=2.1mm]
    	& & & & \gate{R}\gategroup[2,steps=1,style={dashed,rounded
    		corners,draw=red!50,inner xsep=1pt},background]{} & & & & & \gate{R^{-1}}\gategroup[2,steps=1,style={dashed,rounded
    		corners,draw=red!50,inner xsep=1pt},background]{} & & & \midstick[wires=2,brackets=none]{$\equiv$} & & & \midstick[wires=2,brackets=none]{$\equiv$} & & \\
    	& \gate{H} & \gate{B} & & \gate{R^{-1}} & & \gate{B} & \gate{H} & & \gate{R} & & & & \gate{L} & \gate{R} & & \gate{P} &
    \end{quantikz}
    \end{matrix}
\end{align}

We combine the circuits in Eqns.~\eqref{fig:P-I}~and~\eqref{fig:I-P} to compute $P \oplus P$ as in Eq.~\eqref{fig:simpler-P-P}:
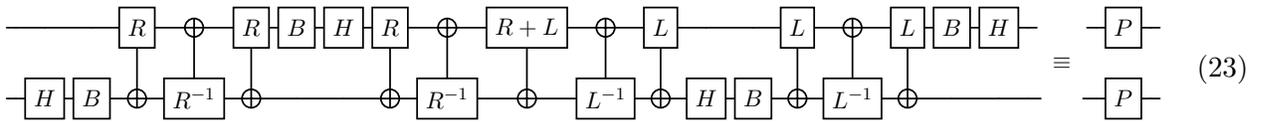
\begin{align}\label{fig:simpler-P-P}
%    \begin{matrix}
%        \includegraphics[scale=0.2894]{figures/P-P-new.png}
%    \end{matrix} \\
    \tikzsetnextfilename{P-P-new}
    \begin{matrix}
   	\begin{tikzpicture}
   	\node[scale=0.82]{
    \begin{quantikz}[row sep={11.5mm,between origins},align equals at=1.5,column sep=1.5mm]
    	& & & & \gate{R} & \targ{}\wire[d][1]{q} & \gate{R} & \gate{B} & \gate{H} & \gate{R} & \targ{}\wire[d][1]{q} & \gate{R+L} & \targ{}\wire[d][1]{q} & \gate{L} & & & \gate{L} & \targ{}\wire[d][1]{q} & \gate{L} & \gate{B} & \gate{H} & & \midstick[2,brackets=none]{$\equiv$} & & \gate{P} & & \\
    	& & \gate{H} & \gate{B} & \targ{}\wire[u][1]{q} & \gate{R^{-1}} & \targ{}\wire[u][1]{q} & & & \targ{}\wire[u][1]{q} & \gate{R^{-1}} & \targ{}\wire[u][1]{q} & \gate{L^{-1}} & \targ{}\wire[u][1]{q} & \gate{H} & \gate{B} & \targ{}\wire[u][1]{q} & \gate{L^{-1}} & \targ{}\wire[u][1]{q} & & & & & & \gate{P} & &
    \end{quantikz}
	};
    \end{tikzpicture}
    \end{matrix}
\end{align}

Note some simplifications in the circuit of Eq.~\eqref{fig:simpler-P-P}.
First, we can eliminate the two pairs of swap gates by moving $H$ and $B$ gates to different wires.
Second, we use the fact illustrated in Eq.~\eqref{fig:R+L} to save one layer of commutative depth.

Finally, we can remove the pair of $B$ gates acting of each of the two $n$-qubit registers, by moving around the control qubits and the target qubits of each \textsf{CNOT} gate.
The result is a circuit of commutative depth 15 for $P \oplus P$.

%------------------------------------------------------------------------------%
\subsection{Prefix sums in commutative depth 16 and size \texorpdfstring{$O(n \log n)$}{O(n log n)}}

From the previous section, we can compute $P \oplus P$ on two $n$-qubit registers for any $n = 2^k -1$ in commutative depth $15$ .
By adding one more layer of commuting \textsf{CNOT} gates, we can compute $P$ for any $n = 2(2^k-1)$ by the circuit construction in Eq.~\eqref{fig:P-final-step}:
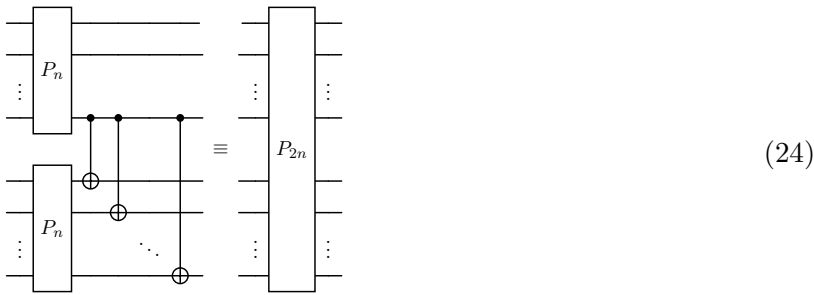
\begin{align}\label{fig:P-final-step}
%\begin{matrix}
%        \includegraphics[scale=0.17]{figures/P-final-step.png}
%    \end{matrix} \\
    \tikzsetnextfilename{P-final-step}
    \begin{matrix}
     \begin{tikzpicture}
    \node[scale=0.7] {
    \begin{quantikz}[row sep={6mm,between origins},column sep=2mm]
    	& & \gate[4]{P_n} & & & & & \midstick[9,brackets=none]{$\equiv$} & & \gate[9]{P_{2n}} & & \\
    	& & & & & & & & & & & \\
    	\setwiretype{n} & \vdots & & & & & & & \vdots & & \vdots & \\
    	& & & \ctrl{2} & \ctrl{3} & & \ctrl{5} & & & & & \\
        \setwiretype{n} & & & & & & & & & & & & \\
    	& & \gate[4]{P_n} & \targ{} & & & & & & & & \\
    	& & & & \targ{} & & & & & & & \\
    	\setwiretype{n} & \vdots & & & & \ddots & & & \vdots & & \vdots & \\
    	& & & & & & \targ{} & & & & & 
    \end{quantikz}
    };
    \end{tikzpicture}
    \end{matrix}
\end{align}

The overall approach can be extended to the case of computing prefix sums for all even $n$, in appendix~\ref{appendix:B}.

%------------------------------------------------------------------------------%

\section{Circuits for arbitrary linear and Clifford operations}

    In this section, we show that, for any even number (denoted as $2n$) of qubits, an arbitrary linear operation has a circuit of commutative depth $11$ (\Cref{cor:main_linear}), and an arbitrary Clifford operation has commutative depth $16$ (\Cref{cor:main_clifford}). Our result are an improvement over \cite{BravyiMaslov+2022}, which achieves commutative depth $18$ for linear operations and $23$ for Clifford operations.
    \subsection{Linear operations}
    We present an implementation of arbitrary linear operations with commutative depth $11$. We begin by arguing that the top left submatrix can be made invertible.
    The following was essentially proven by Hasegawa and Hayashi~\cite{HasegawaH2025}.
    \begin{lemma} 
        \label{lem:invertible_upper_block_2}
        Let $M$ be an $n \times n$ invertible binary matrix and $m < n$.
        Let $A'$ be the principal $m \times m$ submatrix of $M$ in the sense that $M = [\begin{smallmatrix} A' & B' \\ C' & D' \end{smallmatrix}]$.
        Then there exists an $(n-m) \times m$ matrix $X$ such that
        \begin{align}
            \begin{bmatrix}
                A^{'} & B^{'} \\
                C^{'} & D^{'}
            \end{bmatrix}
            \begin{bmatrix}
                I & 0 \\
                X & I
            \end{bmatrix}
            =\begin{bmatrix}
                A & B \\
                C & D
            \end{bmatrix},
        \end{align}
        where $A = A^{'}+B^{'}X$ is invertible.
    \end{lemma}
    \begin{proof} 
        $\begin{bmatrix}
                A^{'} & B^{'}
            \end{bmatrix}$
        has $m$ linear independent columns. Let $a_1^{'},\ldots,a_k^{'},b_1^{'},\ldots,b_{m-k}^{'}$ be the column numbers corresponding to linear independent columns, $a_1,\ldots,a_{m-k}$ be column numbers corresponding to columns in $A^{'}$ that are linearly dependent on $a_1^{'},\ldots,a_k^{'},b_1^{'},\ldots,b_{m-k}^{'}$. There exists a matrix $X$ such that map columns $b_1^{'},\ldots,b_{m-k}^{'}$ in $B^{'}$ to columns $a_1,\ldots,a_{m-k}$ in $B^{'}X$ and other columns to $0$. Thus $A = A^{'}+B^{'}X$ is full rank.
    \end{proof}

	\begin{lemma}[Schur Complement]
		\label{lem:schur}
        Suppose $A$ is an invertible submatrix in the top left corner of a larger square matrix. Then 
		\[
		\begin{bmatrix}
			A & B \\
			C & D
		\end{bmatrix} 
		=
		\begin{bmatrix}
			I & 0 \\
			CA^{-1} & I 
		\end{bmatrix}
		\begin{bmatrix}
			A & 0 \\
			0 & S
		\end{bmatrix}
		\begin{bmatrix}
			I & A^{-1}B \\
			0 & I 
		\end{bmatrix}
		\]
		where $S := D - CA^{-1}B$ is called the \emph{Schur complement}. 
	\end{lemma}
    
    % We take a result Thompson \cite{commutators} that lets us decompose an arbitrary matrix. 
	\begin{theorem}[Thompson \cite{commutators}]
		\label{thm:commutator}
		Every matrix of dimension $n \geq 3$ whose determinant is $1$ can be written as a group commutator $PQP^{-1}Q^{-1}$. 
	\end{theorem}

	\begin{theorem}
        \label{thm:ten_layers}
        Suppose $M = [\begin{smallmatrix} A & B \\ C & D \end{smallmatrix}]$ is a $n \times n$ matrix where $n=2m$ and $A$ is an $m \times m$ invertible submatrix. For $m \geq 3$, there is a circuit for an arbitrary $n \times n$ matrix $M$ that has a commutative depth of $10$, and uses no ancillas.
	\end{theorem}
	\begin{proof}
		Suppose $M$ is $n \times n$, and consider the $m \times m$ blocks. By assumption, the upper left submatrix ($A$) is invertible, so we can apply \Cref{lem:schur}. 
		\begin{equation}
		M = 
		\begin{bmatrix}
			A & B \\
			C & D
		\end{bmatrix} 
		=
		\begin{bmatrix}
			I & 0 \\
			CA^{-1} & I 
		\end{bmatrix}
		\begin{bmatrix}
			A & 0 \\
			0 & S
		\end{bmatrix}
		\begin{bmatrix}
			I & A^{-1}B \\
			0 & I 
		\end{bmatrix}.
		\end{equation}
		In other words, a layer at the beginning and end reduce the problem to a block diagonal matrix. 
		Further decompose the diagonal block matrix as 
		\begin{equation}
		\begin{bmatrix}
			A & 0 \\
			0 & S 
		\end{bmatrix}
		=
		\begin{bmatrix}
			A & 0 \\
			0 & A^{-1} 
		\end{bmatrix}
		\begin{bmatrix}
			I & 0 \\
			0 & AS
		\end{bmatrix}.
		\end{equation}
		We use \Cref{thm:commutator} to write $AS$ as a commutator $PQP^{-1}Q^{-1}$. We have the following trick to implement a commutator: 
		\begin{equation}
			\begin{bmatrix}
				I & 0 \\
				0 & AS 
			\end{bmatrix}
			=
			\begin{bmatrix}
				I & 0 \\
				0 & PQP^{-1}Q^{-1} 
			\end{bmatrix}
			= 
			\begin{bmatrix}
				P^{-1} & 0 \\
				0 & P
			\end{bmatrix}
			\begin{bmatrix}
				PQ^{-1} & 0 \\
				0 & QP^{-1}
			\end{bmatrix}
			\begin{bmatrix}
				Q & 0 \\
				0 & Q^{-1}
			\end{bmatrix}.
		\end{equation}
		Altogether, this means we can write $M$ as 
		\begin{equation}
		\label{eq:main_decomp}
		M = 
		\begin{bmatrix}
			I & 0 \\
			CA^{-1} & I 
		\end{bmatrix}
		\begin{bmatrix}
			A & 0 \\
			0 & A^{-1} 
		\end{bmatrix}
		\begin{bmatrix}
			P^{-1} & 0 \\
			0 & P
		\end{bmatrix}
		\begin{bmatrix}
			PQ^{-1} & 0 \\
			0 & QP^{-1}
		\end{bmatrix}
		\begin{bmatrix}
			Q & 0 \\
			0 & Q^{-1}
		\end{bmatrix}
		\begin{bmatrix}
			I & A^{-1}B \\
			0 & I 
		\end{bmatrix}
		\end{equation}
		where the matrices on either end cost $1$ layer each, and the four matrices in the middle are $3$ layers apiece by \eqref{fig:M-M}. The circuit is as follows.

    \begin{align}\label{fig:p1}
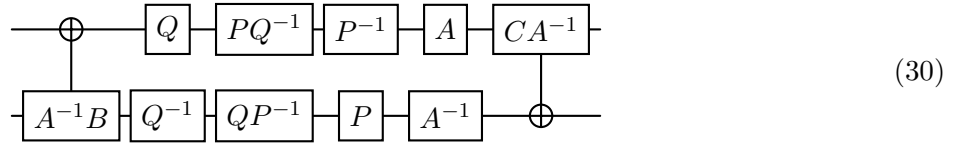
    \begin{matrix}
    \tikzsetnextfilename{c1}
    \begin{quantikz}[row sep={11.5mm,between origins},align equals at=1.5,column sep=1.5mm]
    	& \targ{} & \gate{Q} & \gate{PQ^{-1}} & \gate{P^{-1}} & \gate{A} & \gate{CA^{-1}}\wire[d][1]{q} & \\
    	& \gate{A^{-1}B}\wire[u][1]{q} & \gate{Q^{-1}} & \gate{QP^{-1}} & \gate{P} & \gate{A^{-1}} & \targ{} & 
    \end{quantikz}
    \end{matrix}
    \end{align}
    
    We combine adjacent gates with the same orientation by \Cref{fig:R+L} reducing the commutative depth to $10$. 

    \begin{align}\label{fig:p2}
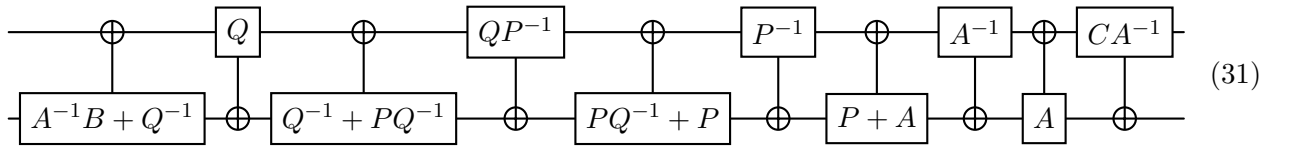
    \begin{matrix}
        \tikzsetnextfilename{p2}
        \begin{quantikz}[row sep={11.5mm,between origins},align equals at=1.5,column sep=1.5mm]
        	& \targ{}\wire[d][1]{q} & \gate{Q} & \targ{}\wire[d][1]{q} & \gate{QP^{-1}} & \targ{}\wire[d][1]{q} & \gate{P^{-1}} & \targ{}\wire[d][1]{q} & \gate{A^{-1}} & \targ{}\wire[d][1]{q} & \gate{CA^{-1}} & \\
        	& \gate{A^{-1}B + Q^{-1}} & \targ{}\wire[u][1]{q} & \gate{Q^{-1} + PQ^{-1}} & \targ{}\wire[u][1]{q} & \gate{PQ^{-1} + P} & \targ{}\wire[u][1]{q} & \gate{P+A} & \targ{}\wire[u][1]{q} & \gate{A} & \targ{}\wire[u][1]{q} & 
        \end{quantikz}
    \end{matrix}  
    \end{align}
    
	\end{proof}

    \begin{corollary}
        \label{cor:main_linear}
        Given an arbitrary matrix $M \in \mathbb F_2^{2m \times 2m}$, there are circuits to compute the linear transformation $M$ in commutative depth $11$.
    \end{corollary}
    \begin{proof}
        By \Cref{lem:invertible_upper_block_2}, there exists a layer of CNOT gates that transforms the matrix into some $M'$ where the top left block is invertible. We then use \Cref{thm:ten_layers} to construct a circuit for $M'$. Finally, having constructed $M'$, we undo the layer of CNOT gates to get $M$. We use $10$ layers for \Cref{thm:ten_layers}, and then one layer for \Cref{lem:invertible_upper_block_2}.
    \end{proof}

    \subsection{Clifford operations}

    The Clifford operations are generated by Hadamard ($H$), Phase ($S$), and CNOT gates, and thus clearly contain linear operations (generated by CNOT) or affine operators (CNOT and $X$) as important subgroups. On the other hand, there are a number of ways to decompose a Clifford operation into the generators as layers of single-qubit gates, CNOT gates, or CZ gates. We use the following decomposition taken from Bravyi and Maslov \cite[Lemma 8]{Bravyi2020HadamardFreeCE}. 
    \begin{theorem}
        \label{thm:decompose_clifford}
        Any Clifford operation may be decomposed into a sequence of 6 layers, $\mathcal{S}$-CNOT-CZ-$\mathcal{S}$-CZ-$\mathcal{S}$, where $\mathcal{S}$ represents a layer of single qubit gates, CNOT is an arbitrary invertible linear operation, and CZ is a layer of CZ gates.
    \end{theorem}
    As an easy corollary of this result and \Cref{cor:main_linear}, there exist constant-commutative-depth circuits for Clifford operations. 
    \begin{corollary}
        \label{cor:main_clifford}
        An arbitrary Clifford operation on $2m$ qubits can be computed by some circuit with commutative depth $16$ and no ancillas.
    \end{corollary}
    \begin{proof}
        Given an arbitrary Clifford operation,     \Cref{thm:decompose_clifford} decomposes it into a linear operation and five other layers. By \Cref{cor:main_linear}, the linear part requires at most $11$ layers. Single-qubit layers are depth $1$ even ordinarily, and the CZ layers have commutative depth $1$ each. Hence, the total is $16$ layers. 
    \end{proof}

    \subsection{Lower bound}

    We also show lower bounds on commutative depth based on counting arguments. The first is about CNOT circuits for linear operations.  
    \begin{theorem}
        \label{thm:lower_bound_linear}
        On sufficiently many bits, there exist linear operations requiring circuits of CNOT gates with commutative depth $4$ or higher.
    \end{theorem}
    \begin{proof}
        The number of invertible linear transformations on $n$ bits is 
        \(
        \prod_{i=0}^{n-1} (2^n - 2^i) = 2^{n^2 - O(1)}. 
        \)
        See the online encyclopedia of integer sequences (OEIS) entry \cite{oeis002884}, or \cite[Appendix A]{clifford_classification}.
        
        On the other hand, for a layer of commuting CNOTs, the bits can be partitioned into a ``control set'' $C$ and ``target set'' $T$ such that all of the CNOTs have their control in $C$ and target in $T$. Hence, the layer computes a transformation of the form $[\begin{smallmatrix} I & 0 \\ X & I \end{smallmatrix}]$ (up to reordering the bits) where $X$ is a $|T| \times |C|$ matrix. The number of entries in this matrix, $|C| \cdot |T|$, is maximized when $|C| = |T| = \tfrac{n}{2}$. 

        Observe that there are at most $2^n$ choices for how to partition $n$ qubits into controls and targets. Then there are at most $2^{n^2/4}$ choices for $X$, the linear transformation from controls to targets implemented by the layer. Thus, there are at most $2^{n^2/4 + n}$ choices for one layer in the commutative depth model. In depth $\leq d$ there are at most 
        \[
        \sum_{i=0}^{d} (2^{n^2/4 + n})^d \leq d \cdot 2^{dn^2/4 + dn}
        \]
        circuits in our commutative depth model. If we consider depth at most $d=3$, then 
        \[
        \text{\# of linear operations} = 2^{n^2 - O(1)} \geq 3 \cdot 2^{3n^2/4 + 3n} \geq \text{\# of depth-3 circuits}
        \]
        and \emph{some} linear operation requires depth $4$ or more. 
    \end{proof}
    
    A similar approach bounds the Clifford operations as well. 
    \begin{theorem}
        \label{thm:lower_bound_clifford}
        On sufficiently many qubits, there exist Clifford operations requiring commutative depth $4$ or more.
    \end{theorem}
    \begin{proof}
        First, we move to a slightly larger gate set. A \emph{generalized CNOT gate} is the two-qubit Clifford gate 
        \[
        C(P, Q) := \tfrac{1}{2}(I \otimes I + P \otimes I + I \otimes Q - P \otimes Q),
        \]
        defined for Pauli operators $P, Q \in \{ X, Y, Z \}$. This includes CNOT ($C(Z,X)$) and CZ ($C(Z,Z)$), and their equivalents in other bases. Two such gates commute if and only if they have the same Pauli on the qubit(s) where they overlap.

        With circuits of generalized CNOT gates, swap gates, and single-qubit gates, we claim a single layer in the commutative depth model has at most $2^{n^2/2 + O(n)}$ possibilities. To see this, note that each qubit can be assigned a Pauli: $X$, $Y$ or $Z$ depending on the generalized CNOT gates which act on it (they must all agree or they do not commute), or $I$ if there is no such gate (i.e., SWAP or single-qubit gates only). Then for each pair of $X$/$Y$/$Z$ qubits, we can choose to have a generalized CNOT gate or not (they are self-inverse), and for each $I$ qubit we can choose from $24$ single-qubit Clifford gates. In other words, 
        \[
        2^{\tfrac{1}{2}(n-i)(n-i-1)} 24^{i}
        \]
        choices where $i$ is the number of $I$ qubits. For sufficiently large $n$, this is maximized when $i = 0$, where it is $2^{n^2/2 + O(n)}$. The choice from $2^{2n}$ possible Paulis for the qubits is in the lower order terms, so there are at most $2^{n^2/2 + O(n)}$ single layers under this gate set in the commutative depth model. 
        
        On the other hand, the number of Clifford operations on $n$ qubits is asymptotically $2^{2n^2 + 3n + O(1)}$. See OEIS \cite{oeis003956} (which includes a superfluous factor of 8 for the phase) or \cite[Appendix A]{clifford_classification}. As before, some operations cannot be depth $3$ because there are not enough depth $3$ circuits for all Clifford operations. 
        \[
        \text{\# of Clifford operations} = 2^{2n^2 + 3n + O(1)} \geq (2^{n^2/2 + O(n)})^{3} \geq \text{\# of depth-3 Clifford circuits.}
        \]
    \end{proof}

    \begin{theorem}
        For any $n$, there exists an $n$-qubit Clifford operation for which any Clifford implementation requires commutative depth $d \geq \frac{n}{5}$ or at least $\frac{n^2}{2(1 + \log_2 d)}$ two-qubit gates. For instance, $O(1)$ commutative depth circuits require $\Omega(n^2)$ gates. 
    \end{theorem}
    \begin{proof}
        As before, the proof is a counting argument. Suppose (toward a contradiction) that all Clifford operations have circuits of commutative depth $d \leq \tfrac{n}{5}$ with at most $s \leq \frac{n^2}{2(1 + \log_2 d)}$ two-qubit gates. We will show that there are not enough circuits for all $2^{2n^2 + 3n + O(1)}$ Clifford operations on $n$ qubits. 
        
        Consider an arbitrary Clifford circuit of commutative depth $d$ with $s$ two-qubit gates. First, we divide each layer into a single-qubit layer followed by a two-qubit layer---recall that we may reorder gates within a layer arbitrarily since they commute. 

        Next, the two-qubit Clifford generators are generalized CNOT gates that, as established in the previous theorem, commute if and only if they share the same Pauli on the qubit(s) where they overlap. If we conjugate by an appropriate single-qubit gate on the qubits with $X$-controls or $Y$-controls, we can make them all $Z$-controls. Without loss of generality, the two-qubit gates are all CZs since the single-qubit gates can be absorbed into the layer of single-qubit gates before and after each two-qubit layer. Note that this creates a $d+1$'st layer of single-qubit gates at the end of the circuit, and the number of two-qubit gates is preserved. 
        
        There are $24^{nd}$ configurations of the main $d$ single-qubit layers, since there are $24$ single-qubit Clifford gates and $nd$ sites where they are applied. The final layer contributes an additional $3^{n}$ configurations, since each qubit is conjugated by one of three gates: identity, $H$, or $HS$. We get the following bound on the number of single-qubit gate configurations in our circuit:
        \[
        \text{\# single-qubit gate config.} \leq 24^{nd} 3^{n} \leq 2^{4.585nd+1.585n} \leq 2^{n^2 + 2n},
        \]
        using the fact that $d \leq \tfrac{n}{5}$.

        In the $d$ two-qubit layers, there are $\binom{n}{2}$ positions for a CZ per layer, or $N := d \cdot \binom{n}{2}$ total. It follows that there are $\binom{N}{k}$ configurations of the two-qubit layers having a total of $k$ CZ gates. Now let $k$ range from $0$ up to $s$, the size of the circuit counting two-qubit gates only. We can upper bound this sum with
        \[
        \sum_{k=0}^{s} \binom{N}{k} \leq \sum_{k=0}^{s} \frac{N^k}{k!} = \sum_{k=0}^{s} \left( \frac{N}{s} \right)^k \frac{s^k}{k!} \leq  \left( \frac{N}{s} \right)^s \left( \sum_{k=0}^{s} \frac{s^k}{k!} \right) \leq \left( \frac{Ne}{s} \right)^{s},
        \]
        where we have used that $s \leq N$. It is not hard to check that this function is increasing for all $s < N$, so we can bound the number of two-qubit gate configurations by substituting the upper bound $\frac{n^2}{1 + \log_2 d}$ for $s$:
        \begin{align*}
        \text{\# two-qubit gate config.} 
        &\leq \left( \frac{Ne}{s} \right)^s \\
        &\leq \left( \frac{2\binom{n}{2} d(1 + \log_2 d)e}{n^2} \right)^{\frac{n^2}{1 + \log d}} \\
        &\leq \left( 4d^2 \right)^{\frac{n^2}{2(1 + \log d)}} && \text{since $1 + \log_2 d \leq d$, $e \leq 4$, and $2\binom{n}{2} \leq n^2$,} \\
        &= 2^{n^2} &&\text{since $4d^2 = 2^{2(1 + \log_2 d)}$.}
        \end{align*}

        We showed that single-qubit gates contribute $\leq 2^{n^2 + 2n}$ configurations, and two-qubit gates contribute $\leq 2^{n^2}$, so there are at most $2^{n^2 + 2n}$ circuits of depth $d \leq \tfrac{n}{5}$ with at most $s \leq \tfrac{n^2}{2(1 + \log_2 d)}$ two-qubit gates. There are $2^{2n^2 + 3n + O(1)}$ Clifford operations on $n$ qubits, however, contradicting our initial assumption and finishing the proof. 
    \end{proof}
\section{Acknowledgments}

We would like to thank Atsuya Hasegawa and Koyo Hayashi for sharing their version of Lemma~\ref{lem:invertible_upper_block_2} \cite{HasegawaH2025} and other helpful discussions about this work, and Crystal Senko for pointing us to Refs.~\cite{FiggattO+2019,GrzesiakB+2020}.

%%%%%%%%%%%%%%%%%%%%%%%%%%%%%%%%%%

\bibliography{ref}
\bibliographystyle{plain}

%%%%%%%%%%%%%%%%%%%%%%%%%%%%%%%%%%
\ifbool{FULL}{

%------------------------------------------------------------------------------%
\appendix

%------------------------------------------------------------------------------%
\section{Analysis of the density of 1s in \texorpdfstring{$L$}{L}, \texorpdfstring{$R$}{R}, \texorpdfstring{$L^{-1}$}{inverse L}, and \texorpdfstring{$R^{-1}$}{inverse R}}\label{appendix:A}
%------------------------------------------------------------------------------%

For the case where $n=15$, $L$ is the circuit in \Cref{fig:LF-L-R}(a), and the binary matrix associated with this circuit is
\begin{align}
    L_{15} = 
    \begin{bmatrix}
        1 &   &   &   &   &   &   &   &   &   &   &   &   &   &   \\
        1 & 1 &   &   &   &   &   &   &   &   &   &   &   &   &   \\
        0 & 0 & 1 &   &   &   &   &   &   &   &   &   &   &   &   \\
        1 & 1 & 1 & 1 &   &   &   &   &   &   &   &   &   &   &   \\
        0 & 0 & 0 & 0 & 1 &   &   &   &   &   &   &   &   &   &   \\
        0 & 0 & 0 & 0 & 1 & 1 &   &   &   &   &   &   &   &   &   \\
        0 & 0 & 0 & 0 & 0 & 0 & 1 &   &   &   &   &   &   &   &   \\
        1 & 1 & 1 & 1 & 1 & 1 & 1 & 1 &   &   &   &   &   &   &   \\
        0 & 0 & 0 & 0 & 0 & 0 & 0 & 0 & 1 &   &   &   &   &   &   \\
        0 & 0 & 0 & 0 & 0 & 0 & 0 & 0 & 1 & 1 &   &   &   &   &   \\
        0 & 0 & 0 & 0 & 0 & 0 & 0 & 0 & 0 & 0 & 1 &   &   &   &   \\
        0 & 0 & 0 & 0 & 0 & 0 & 0 & 0 & 1 & 1 & 1 & 1 &   &   &   \\
        0 & 0 & 0 & 0 & 0 & 0 & 0 & 0 & 0 & 0 & 0 & 0 & 1 &   &   \\
        0 & 0 & 0 & 0 & 0 & 0 & 0 & 0 & 0 & 0 & 0 & 0 & 1 & 1 &   \\
        0 & 0 & 0 & 0 & 0 & 0 & 0 & 0 & 0 & 0 & 0 & 0 & 0 & 0 & 1 \\
    \end{bmatrix}.
\end{align}
For any $n = 2^k - 1$, it is straightforward to deduce that $L_n$ has the recursive structure  
\begin{align}
\mbox{\large $L_{2n+1}$} = 
\begin{bmatrix}
    & \!\!\!\!\!\begin{smallmatrix}\ \end{smallmatrix} & \\[-3.0mm]
    \, \mbox{\large $L_n$} & & \\
    \,\begin{smallmatrix}1\, & \cdots & 1\end{smallmatrix} & \!\!\!\!\!\begin{smallmatrix}1\end{smallmatrix} & \\[-1.5mm]
    & \!\!\!\!\!\begin{smallmatrix}0\end{smallmatrix} \\[-0.8mm]
    \ \mbox{\LARGE $0$} & \!\!\!\!\!\begin{smallmatrix}{:} \\[1.3mm] 0\end{smallmatrix} & \!\!\! \mbox{\large $L_n$}\,\, \\[2mm]
    \end{bmatrix}.
\end{align}
Therefore, the density of 1s in $L_n$, which we denote by $d_n$, satisfies the recurrence
\begin{align}
    d_{2n+1} = 2d_n + n+1,
\end{align}
which implies $d_n = O(n\log n)$.

The binary matrix associated with $L^{-1}$ for the case $n=15$ is
\begin{align}
    L^{-1}_{15} = 
    \begin{bmatrix}
        1 &   &   &   &   &   &   &   &   &   &   &   &   &   &   \\
        1 & 1 &   &   &   &   &   &   &   &   &   &   &   &   &   \\
        0 & 0 & 1 &   &   &   &   &   &   &   &   &   &   &   &   \\
        0 & 1 & 1 & 1 &   &   &   &   &   &   &   &   &   &   &   \\
        0 & 0 & 0 & 0 & 1 &   &   &   &   &   &   &   &   &   &   \\
        0 & 0 & 0 & 0 & 1 & 1 &   &   &   &   &   &   &   &   &   \\
        0 & 0 & 0 & 0 & 0 & 0 & 1 &   &   &   &   &   &   &   &   \\
        0 & 0 & 0 & 1 & 0 & 1 & 1 & 1 &   &   &   &   &   &   &   \\
        0 & 0 & 0 & 0 & 0 & 0 & 0 & 0 & 1 &   &   &   &   &   &   \\
        0 & 0 & 0 & 0 & 0 & 0 & 0 & 0 & 1 & 1 &   &   &   &   &   \\
        0 & 0 & 0 & 0 & 0 & 0 & 0 & 0 & 0 & 0 & 1 &   &   &   &   \\
        0 & 0 & 0 & 0 & 0 & 0 & 0 & 0 & 0 & 1 & 1 & 1 &   &   &   \\
        0 & 0 & 0 & 0 & 0 & 0 & 0 & 0 & 0 & 0 & 0 & 0 & 1 &   &   \\
        0 & 0 & 0 & 0 & 0 & 0 & 0 & 0 & 0 & 0 & 0 & 0 & 1 & 1 &   \\
        0 & 0 & 0 & 0 & 0 & 0 & 0 & 0 & 0 & 0 & 0 & 0 & 0 & 0 & 1 \\
    \end{bmatrix},
\end{align}
and, for any $n = 2^k - 1$, has the recursive structure
\begin{align}
\mbox{\large $L^{-1}_{2n+1}$} = 
\begin{bmatrix}
    & \!\!\!\!\!\begin{smallmatrix}\ \end{smallmatrix} & \\[-3.0mm]
    \, \mbox{\large $L^{-1}_n$} & & \\[1mm]
    \,\begin{smallmatrix}b_n & \cdots & b_1\end{smallmatrix} & \!\!\!\!\!\begin{smallmatrix}1\end{smallmatrix} & \\[-1.5mm]
    & \!\!\!\!\!\begin{smallmatrix}0\end{smallmatrix} \\[-0.8mm]
    \ \mbox{\LARGE $0$} & \!\!\!\!\!\begin{smallmatrix}: \\[1.3mm] 0\end{smallmatrix} & \!\!\! \mbox{\large $L^{-1}_n$}\,\, \\[2mm]
    \end{bmatrix},
\end{align}
where
\begin{align}
    b_j =
    \begin{cases}
        1 & \mbox{if $j$ is a power of 2} \\
        0 & \mbox{otherwise.}
    \end{cases}
\end{align}
Therefore, the density of 1s in $L^{-1}_n$, which we denote by $\tilde{d}_n$, satisfies the recurrence
\begin{align}
    \tilde{d}_{2n+1} = 2\tilde{d}_n + O(\log n),
\end{align}
which implies $\tilde{d}_n = O(n)$.

The weight of $R$ and $R^{-1}$ can be deduced from the weights of $L$ and $L^{-1}$ on account of the following lemma.

\begin{lemma}
    For any invertible binary $n\times n$ matrix $M$, it holds that $H^{\otimes n} B M^{-1} B H^{\otimes n}$ is the anti-transpose of $M$ (where \emph{anti-transpose} of a square matrix is the matrix flipped along the anti-diagonal, which is the diagonal from the bottom-left corner to the top-right corner).
\end{lemma}

Combining this with Eqns.~\eqref{eq:R-in-terms-of-L}~and~\eqref{eq:L-in-terms-of-R}, we can deduce that the weight of $R$ is $O(n \log n)$ and the weight of $R^{-1}$ is $O(n)$.

%------------------------------------------------------------------------------%
\section{Prefix Sum for arbitrary even \texorpdfstring{$n$}{n}}\label{appendix:B}

In section~\ref{sec:prefix-sums}, we showed how to compute Prefix Sums for $n = 2(2^k-1)$.
This is easily extended to $n = 2m$ for any odd $m$ by observing that, if the Ladner-Fischer circuits in Eq.~\eqref{fig:LF-L-R} are pruned by an equal number of qubits from the top and bottom then the result also correctly computes the Prefix Sums for the smaller number of qubits.

For example, for the circuit in Eq.~\eqref{fig:LF-L-R}, if the first three qubits and the last three qubits are pruned then the result is a 9-qubit circuit 
%\begin{align}\label{fig:LF-L-R-9-qubits}
%\begin{matrix}
%     \includegraphics[scale=0.18]{figures/LF-for-9-qubits.png}\\
%     \hspace*{-20mm}\mbox{(a) Ladner-Fischer circuits for $n=9$}\hspace*{8mm}
%     \mbox{(b) $L$}\hspace*{23mm}\mbox{(c) $R$}
%\end{matrix}
%\end{align}

%%% [row sep={6mm,between origins},column sep=3mm,align equals at=8]

\begin{equation}\label{fig:LF-L-R-9-qubits}
\begin{matrix}
	\tikzsetnextfilename{LF-for-9-qubits-tikz}
	\begin{tikzpicture}
		\node[scale=0.65] {
        \begin{quantikz}[row sep={6mm,between origins},column sep=3mm]
		& & & \ctrl{4} & & \ctrl{2} & \ctrl{1} & \\
		& \ctrl{1} & & & & & \targ{} & \\
		& \targ{} & \ctrl{2} & & & \targ{} & \ctrl{1} & \\
		& \ctrl{1} & & & & & \targ{} & \\
		& \targ{}& \targ{} & \targ{} & \ctrl{4} & \ctrl{2} & \ctrl{1} & \\
		& \ctrl{1} & & & & & \targ{} & \\
		& \targ{} & \ctrl{2} & & & \targ{} & \ctrl{1} & \\
		& \ctrl{1} & & & & & \targ{} & \\
		& \targ{}& \targ{} & & \targ{} & & &
	\end{quantikz}
    };
    \end{tikzpicture}
	&
	\tikzsetnextfilename{LF-for-9-qubits-L}
	\begin{tikzpicture}
		\node[scale=0.65] {
        \begin{quantikz}[row sep={6mm,between origins},column sep=3mm]
		& & & \ctrl{4} & \\
		& \ctrl{1} & & & \\
		& \targ{} & \ctrl{2} & & \\
		& \ctrl{1} & & & \\
		& \targ{} & \targ{} & \targ{} &  \\
		& \ctrl{1} & & & \\
		& \targ{} & \ctrl{2} & & \\
		& \ctrl{1} & & & \\
		& \targ{} & \targ{} & &
	\end{quantikz}
    };
    \end{tikzpicture}
    &
	\tikzsetnextfilename{LF-for-9-qubits-R}
	\begin{tikzpicture}
		\node[scale=0.65] {
        \begin{quantikz}[row sep={6mm,between origins},column sep=3mm]
        & & \ctrl{2} & \ctrl{1} & \\
		& & & \targ{} & \\
		& & \targ{} & \ctrl{1} & \\
		& & & \targ{} & \\
		& \ctrl{4} & \ctrl{2} & \ctrl{1} & \\
		& & & \targ{} & \\
		& & \targ{} & \ctrl{1} & \\
		& & & \targ{} & \\
		& \targ{} & & &
	\end{quantikz} 
        };
    \end{tikzpicture} \\
    \text{(a) Pruned Ladner-Fischer ($n = 9$)} & \text{\ \ \ \ \ \ (b) $L$\ \ \ \ \ \  \ } & \text{\ \ \ \ \ \ (c) $R$\ \ \ \ \ \ \ }
\end{matrix}
\end{equation}
which correctly computes the 9-qubit instance of Prefix Sums and which exhibits the same structural properties that are used in section~\ref{sec:prefix-sums}.

This gives us Prefix Sums for all even numbers of the form $2(2k+1) = 4k+2$; however, this does not capture Prefix Sums for the case of even numbers of the form $4k$.

To capture the cases where $n = 4k$, we can add two qubits to the circuit in Eq.~\eqref{fig:simpler-P-P} without increasing the depth as follows.
Start with the $n = 2(2k+1)$ circuit in Eq.~\eqref{fig:simpler-P-P} and add one qubit to the beginning and one qubit to the end.
Add one \textsf{CNOT} gate to the beginning (in parallel with the first $H$ layer) and one \textsf{CNOT} to the end (in parallel with the last $H$ layer), as shown in the fine-grained depiction of the beginning and the end of the circuit in Eq.~\eqref{fig:simpler-P-P}:
%\begin{align}\label{fig:detail-left}
%\begin{matrix}
%     \includegraphics[scale=0.2]{figures/detail.png}
%\end{matrix}
%\end{align}

%\begin{quantikz}[row sep={6mm,between origins},column sep=2mm]
%    & \ctrl{1} & \midstick{\cdots} & 
%\end{quantikz}

\begin{align}\label{fig:detail-left}
\begin{matrix}
    \tikzsetnextfilename{detail-new}
\begin{tikzpicture}
\node[scale=0.7]{
\begin{quantikz}[row sep={2.5mm,between origins},column sep=4mm]
	& \ctrl{2} & & \midstick{\large\ \ \ $\cdots$\ \ \ } & & & \\
    \setwiretype{n} & & & & \\
    %\lstick[5]{$2k+1$} 
    & \targ{} & & \midstick[5,brackets=none]{\large\ \ \ $\cdots$\ \ \ } & \gate[5]{B}\hphantom{\hspace{1cm}} & \gate[5]{H}\hphantom{\hspace{1cm}} & \\
    & & & & & & \\
    & & & & & & \\
    & & & & & & \\
    & & & & & & \\
    \setwiretype{n} & & & & & & \\
    \setwiretype{n} & & & & & & \\
    %\lstick[5]{$2k+1$} 
    & \gate[5]{H}\hphantom{\hspace{1cm}} & \gate[5]{B}\hphantom{\hspace{1cm}} & \midstick[5,brackets=none]{\large\ \ \ $\cdots$\ \ \ } & & & \\
    & & & & & & \\
    & & & & & & \\
    & & & & & & \\
    & & & & & \ctrl{2} & \\
    \setwiretype{n} & & & & \\
    & & & \midstick[5,brackets=none]{\large\ \ \ $\cdots$\ \ \ } & & \targ{} & \\
\end{quantikz}
} ;
\end{tikzpicture} \\
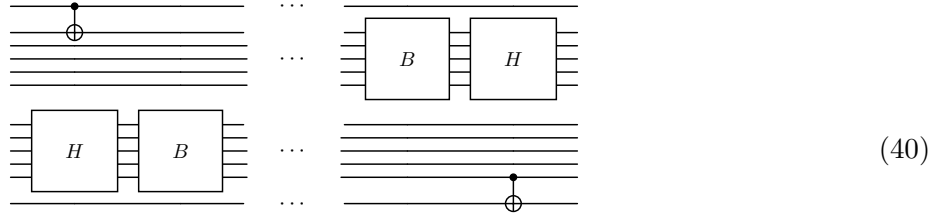
\text{Detailed depiction of the beginning and the end of circuit of Eq.~\eqref{fig:simpler-P-P}}\\
\text{modified with two additional qubits\hspace*{56mm}}
\end{matrix}
\end{align}
}{}
%%%%%%%%%%%%%%%%%%%%%%%%%%%%%%%%%%

\end{document}